\documentclass[11pt,a4paper]{article}
\usepackage[utf8]{inputenc}
\usepackage[T1]{fontenc}
\usepackage[english]{babel}
\usepackage{lmodern}
\usepackage{microtype}
\usepackage[a4paper,top=2.5cm,bottom=2.5cm,left=2.7cm,right=2.7cm,footskip=1.2cm]{geometry}
\usepackage{amsmath,amssymb,amsthm,mathtools,bm}
\usepackage{graphicx,booktabs,array,tabularx,multirow,makecell,longtable}
\usepackage{caption,subcaption}
\usepackage{xcolor,enumitem,csquotes}
\usepackage{tikz}
\usetikzlibrary{arrows.meta,positioning,fit,calc,shapes.geometric,decorations.pathreplacing}
\usepackage{siunitx}
\newcolumntype{L}[1]{>{\raggedright\arraybackslash}p{#1}}
\newcolumntype{Y}{>{\raggedright\arraybackslash}X}
\usepackage[backend=biber,style=authoryear-comp,natbib=true,url=true,doi=true,isbn=false,maxbibnames=99,maxcitenames=2,uniquelist=false,sorting=nyt]{biblatex}
\usepackage[colorlinks=true,linkcolor=black,citecolor=black,urlcolor=black]{hyperref}
\usepackage[noabbrev,capitalize]{cleveref}
\theoremstyle{plain}
\newtheorem{proposition}{Proposition}

\newtheorem{corollary}[proposition]{Corollary}
\theoremstyle{definition}
\newtheorem{definition}{Definition}
\newtheorem{assumption}{Assumption}

\theoremstyle{remark}
\newtheorem*{remark}{Remark}
\crefname{assumption}{Assumption}{Assumptions}
\Crefname{assumption}{Assumption}{Assumptions}
\crefname{definition}{Definition}{Definitions}
\Crefname{definition}{Definition}{Definitions}
\crefname{criterion}{Criterion}{Criteria}
\Crefname{criterion}{Criterion}{Criteria}
\setlist{nosep,leftmargin=*}
\definecolor{seriesblue}{RGB}{31,78,121}
\definecolor{seriesgray}{RGB}{242,244,247}
\tikzset{seriesbox/.style={draw=black!70,rounded corners=2pt,align=center,inner sep=5pt,minimum height=9mm,fill=seriesgray},seriesarrow/.style={-{Latex[length=2.2mm]},thick,draw=black!75}}

\title{Resolution-Aware Perpetual Futures on Binary Prediction Markets:\\
Failure Modes and Mechanical Stress Tests Using Polymarket Data}
\author{Maksym Nechepurenko\thanks{Founder and Director of Research, ForesightFlow, the Research Department of Devnull FZCO, Dubai, United Arab Emirates. Email: \texttt{maksym@devnull.ae}.}}
\date{July 19, 2026}

\begin{document}
\maketitle
\begin{abstract}
We study whether crypto-style perpetual-futures mechanics can be carried over to a binary event claim that ultimately pays either zero or one. The paper is not a product proposal. It identifies non-portability conditions and applies a mechanical stress test to candidate controls on observed Polymarket paths.

The formal account identity is sharper than in the first version. A synthetic long entered at price $p_0$ with collateral equal to $xp_0/L$ has terminal equity $-xp_0(1-1/L)$ when the claim pays zero, provided the position survives to resolution without a top-up, close, or conversion. The result is independent of the pre-resolution mark: mark-to-market gains and the terminal jump telescope to the entry-to-terminal payoff. Hence any leverage $L>1$ leaves an adverse-outcome account shortfall under those conditions. A second result gives a funding trilemma. Basis-only funding vanishes for a fixed relative error near a boundary, whereas a uniformly relative-basis correction becomes unbounded and violates payer-solvency or participation constraints. Funding alone therefore cannot provide uniform boundary control.

The empirical layer uses a seven-day Polymarket window from Tuesday 21 April through Monday 27 April 2026. A stratified enrichment pipeline processes 61,087 candidate markets; 13,298 pass the stated adequacy gates. Two structural diagnostics pass, but three of five pre-specified materiality tests fail. Dynamic margin and leverage compression liquidate more observed paths than the static baseline, pooled drawdown falls by only 5.1 percent, and a staged halt lowers final-hour liquidations mechanically while leaving terminal account-shortfall incidence slightly worse. The contribution is the corrected non-portability analysis, a reusable mechanical-replay design, and a set of negative design lessons. The results do not establish equilibrium performance or deployment safety.
\end{abstract}

\noindent\textbf{Keywords:} prediction markets; perpetual futures; event contracts; terminal settlement; margin; funding; mechanical stress testing; Polymarket.


\clearpage
\section{Introduction}
\label{sec:intro}

A binary outcome claim is simple when traded as spot. A trader pays the premium, receives a claim that pays either zero or one, and cannot lose more than the amount paid. The balance sheet changes when the same directional exposure is represented by a leveraged synthetic contract. The claim still resolves at a finite event time, but the trader posts only a fraction of the adverse-outcome loss. If the exposure remains open until resolution, the missing collateral must be supplied by a top-up, a pre-finality close, a conversion into a funded claim, or a loss-bearing third party.

\subsection{A plain account example}

Suppose a trader takes a synthetic long position of $x=1{,}000$ units at an entry price $p_0=0.60$ with leverage $L=5$. The initial collateral is $xp_0/L=120$. If the pre-resolution mark rises to $0.90$, the account shows a mark-to-market gain of $300$ and equity of $420$. If the event then resolves NO, the final mark change is $-0.90$ per unit, or $-900$. Terminal equity is therefore $420-900=-480$. The same result follows directly from entry to terminal: initial collateral $120$ minus the adverse-outcome loss $600$ equals $-480$.

The pre-resolution mark affects when the risk is visible and whether a liquidation rule acts. It does not change the entry-to-terminal account identity. This distinction is central to the revision.

\subsection{Research question}

The paper asks a falsification question rather than an advocacy question:

\begin{quote}
Under what conditions do standard perpetual-futures controls fail when the reference claim has a finite event horizon and a discrete terminal payoff, and how do candidate resolution-aware controls behave mechanically on observed event-market paths?
\end{quote}

The question is narrower than whether leverage on event claims is desirable. A dated future, a fully funded spot claim, a financed physical claim, and a synthetic perpetual overlay are different balance sheets. The paper compares them explicitly in \cref{sec:instrument-choice}. The reason to study a perpetual overlay is not that it is the natural or preferred instrument. It is that perpetual-style margin and funding infrastructure is widely reused in digital-asset markets, making a naive transplant a plausible engineering choice whose failure modes should be understood before deployment.

\subsection{Main findings}

The revision has three central findings.

First, terminal solvency is governed by the total entry-to-terminal payoff, not by the final jump considered in isolation. \Cref{prop:terminal-equity} shows that, under a stated no-top-up/no-close condition, any leverage above one produces adverse-outcome account shortfall for a synthetic long carried to a zero payout. A symmetric result holds for a short carried to a one payout. The theorem is a balance-sheet identity, not an empirical frequency claim.

Second, the boundary problem cannot be solved by allowing funding to diverge. \Cref{prop:funding-trilemma} shows that a bounded, economically payable transfer cannot impose uniform pressure proportional to relative basis arbitrarily close to zero or one. An uncapped relative-basis correction restores the mathematical penalty only by violating payer-solvency or participation constraints. A feasible design must clip transfers, floor the denominator, restrict leverage/listing near the boundary, or accept that funding does not provide uniform convergence.

Third, the observed-path replay is predominantly negative. The data confirm large terminal jumps and boundary-region microstructure differences. The candidate control family nevertheless misses three of five pre-specified materiality floors. Dynamic margin and leverage compression pre-empt recoveries on many observed paths; the halt removes a liquidation window but does not remove the terminal balance-sheet loss; and the boundary-funding component has little incremental effect in the replayed sample.

\subsection{Contribution and interpretation}

The contribution is therefore not a validated engine. It is:

\begin{enumerate}[label=(\roman*)]
  \item a corrected account-level non-portability result;
  \item an economically constrained boundary-funding result;
  \item a large-scale event-market microstructure characterization for derivative design;
  \item a reusable mechanical-replay protocol with pre-specified thresholds and explicit sample gates; and
  \item negative evidence identifying which candidate controls do not solve the dominant loss channel on the observed paths.
\end{enumerate}

The replay keeps historical prices, depth, order placement, open interest proxies, and trader behavior fixed while changing engine rules. It is therefore a mechanical stress test. It does not identify an equilibrium in which traders, arbitrageurs, and liquidity suppliers respond to the presence of the derivative.

\subsection{Organization}

\Cref{sec:instrument-choice} distinguishes the instrument classes. \Cref{sec:account-mechanics} defines the account identity. \Cref{sec:data-design} describes the data and replay. \Cref{sec:nonportability} proves the two structural results. \Cref{sec:candidate-controls} defines the candidate controls and their economic limits. \Cref{sec:results} reports the empirical results. \Cref{sec:interpretation} states the design implications and limitations.

\section{Instrument Choice and Related Design Context}
\label{sec:instrument-choice}

Prediction-market prices aggregate information imperfectly and need not equal objective probabilities because risk preferences, wealth, fees, and market structure enter the price \citep{wolfers_zitzewitz_2004,manski_2006}. The present paper treats the event-market price as a traded reference value, not as truth.

Perpetual futures are non-expiring leveraged contracts whose mark is ordinarily tethered to a continuously tradable spot index through periodic funding \citep{he_manela_schwert_2024_perpetuals,dai_li_yang_2025_arbitrage_perpetuals,zhang_2026_funding_rate_mechanism}. A binary event claim is different: the reference market has a finite economically meaningful horizon and a terminal payoff. Calling the overlay ``perpetual'' therefore describes its pre-resolution trading and funding convention, not an ability to remain open after the event is final.

\subsection{Four distinct balance sheets}

\begin{table}[ht]
\centering
\caption{Instrument alternatives. A fixed maturity removes ambiguity about when a contract ends, but it does not by itself solve collateral insufficiency.}
\label{tab:instrument-comparison}
\begin{tabularx}{\textwidth}{L{3.1cm}YYYY}
\toprule
Structure & Asset or obligation & Maturity & Principal terminal risk & Role in this paper \\
\midrule
Spot outcome claim & Fully paid YES/NO token & Event resolution & Holder can lose the premium, but creates no derivative-account shortfall & Empirical reference market \\
Dated cash-settled future & Bilateral cash obligation & Fixed expiry & Leveraged adverse settlement can exceed posted collateral & Alternative with explicit maturity \\
Synthetic event-linked perpetual & Margin account referencing event price & Delist, settle, convert, or roll at event finality & Same adverse terminal account identity if exposure survives & Object stress-tested here \\
Physically backed financed claim & Real outcome token plus separate loan & Loan matures before or at claim finality & Credit risk moves to execution, collateral control, and lender recovery & Alternative architecture, outside empirical scope \\
\bottomrule
\end{tabularx}
\end{table}

A dated future is not automatically safe. If it is cash-settled and leveraged, the same terminal account identity applies at expiry. A fully funded spot claim avoids derivative bad debt because the premium is paid in full. A physically backed financed claim can preserve leverage before finality, but then the relevant object is a loan secured by actual outcome tokens; lender safety depends on control, execution, and repayment rather than on perpetual funding.

\subsection{What is specific to binary markets}

Two features should be separated.

The first is general to finite-horizon leveraged claims: a terminal payoff can exceed posted collateral, especially when trading closes before the account can be adjusted. Similar risk can arise for real-valued terminal variables, dated derivatives, and contracts with discontinuous cash settlement.

The second is specific to normalized probability-like underlyings. Prices lie in $[0,1]$, the adverse terminal loss has a transparent upper bound, and relative mark--index errors become singular near zero or one. The funding result in \cref{prop:funding-trilemma} relies on this boundary geometry. The collateral identity in \cref{prop:terminal-equity} relies more broadly on finite terminal settlement and under-collateralized exposure.

\subsection{Empirical context}

Polymarket is a central limit order book (CLOB) venue for collateralized event claims. Its high-frequency archive permits direct measurement of spreads, depth, price changes, and final event-market observations. Existing work documents the information-aggregation properties of prediction markets and, more recently, the microstructure of CLOB-based venues \citep{dubach_2026_polymarket_anatomy,tsang_yang_2026_polymarket_anatomy,rahman_alchami_clark_2025_sok_depms}. The present paper uses those observations to test derivative controls mechanically; it does not infer that an event-linked derivative would leave the spot venue unchanged.

\section{Contract and Account Mechanics}
\label{sec:account-mechanics}

\subsection{Event claim and synthetic exposure}

Let $Y\in\{0,1\}$ be the terminal payout of a YES claim at resolution time $\tau$. Let $p_t\in[0,1]$ be the reference event-market price before resolution and $q_t$ the derivative mark. A synthetic long of size $x>0$ entered at $q_0=p_0$ has terminal profit and loss
\begin{equation}
  \Pi_\tau=x(Y-p_0).
  \label{eq:terminal-pnl}
\end{equation}
This identity does not require the mark to track $p_t$ perfectly between entry and resolution.

Let $C_0$ be collateral posted at entry. Define premium-based leverage for the long by
\begin{equation}
  L=\frac{x p_0}{C_0}, \qquad C_0=\frac{x p_0}{L}.
  \label{eq:long-leverage}
\end{equation}
The account equity immediately after terminal settlement, before any insurance or socialized recovery, is
\begin{equation}
  E_\tau=C_0+x(Y-p_0).
  \label{eq:terminal-equity-general}
\end{equation}

\subsection{Pre-resolution mark-to-market and telescoping}

Suppose the account recognizes mark-to-market profit and loss at $q_{\tau^-}$. Immediately before resolution,
\begin{equation}
  E_{\tau^-}=C_0+x(q_{\tau^-}-p_0).
\end{equation}
The final incremental change is $x(Y-q_{\tau^-})$. Hence
\begin{equation}
  E_{\tau^-}+x(Y-q_{\tau^-})=C_0+x(Y-p_0)=E_\tau.
  \label{eq:telescoping-equity}
\end{equation}
The mark path affects liquidation timing, margin calls, and the amount of variation margin exchanged. It does not alter terminal equity when all transfers are accounted for.

\subsection{Account shortfall and exchange loss}

\begin{definition}[Account shortfall]
The terminal account shortfall is
\begin{equation}
  S_\tau=[-E_\tau]^+.
\end{equation}
\end{definition}

Account shortfall is not identical to exchange bad debt. Exchange loss depends on collateral recovery, liquidation proceeds, guarantees, insurance funds, and loss-sharing rules. The empirical replay's ``bad-debt frequency'' is therefore relabeled in this revision as terminal account-shortfall incidence unless the full recovery waterfall is explicitly modeled.

\subsection{Mechanical replay}

Let $X^{\mathrm{obs}}$ denote the observed spot-market path and $g_c$ a deterministic engine configuration. The replay output is
\begin{equation}
  Z_c=g_c(X^{\mathrm{obs}}).
  \label{eq:mechanical-replay}
\end{equation}
It does not estimate the endogenous path $X(c)$ that would arise if configuration $c$ were deployed. In particular, it does not model changes in liquidity supply, arbitrage, entry, exit, open interest, or strategic response.

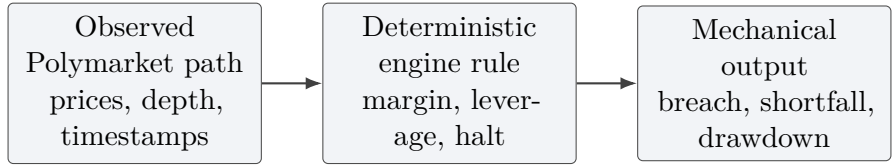
\begin{figure}[ht]
\centering
\begin{tikzpicture}[node distance=12mm and 8mm]
  \node[seriesbox,text width=3.0cm] (obs) {Observed Polymarket path\\prices, depth, timestamps};
  \node[seriesbox,text width=3.0cm,right=of obs] (rule) {Deterministic engine rule\\margin, leverage, halt};
  \node[seriesbox,text width=3.0cm,right=of rule] (out) {Mechanical output\\breach, shortfall, drawdown};
  \draw[seriesarrow] (obs) -- (rule);
  \draw[seriesarrow] (rule) -- (out);
  \node[below=8mm of rule,align=center,text width=10cm] {Not identified: the counterfactual market path after traders and liquidity suppliers respond to the rule.};
\end{tikzpicture}
\caption{Interpretation of the replay. It is a deterministic stress test on fixed historical paths, not an equilibrium deployment simulation.}
\label{fig:replay-interpretation}
\end{figure}

\section{Data and Empirical Design}
\label{sec:data-design}

\subsection{Window and sources}

The empirical window contains 168 contiguous hours from Tuesday 21 April 2026 at 00:00 UTC through Monday 27 April 2026 at 23:59 UTC.

The primary source is the PMXT v2 archive of Polymarket CLOB events. Market metadata are joined from the Gamma interface and resolution information from the venue's oracle infrastructure. The raw archive contains approximately 13.69 billion events and 110,828 distinct observed market identifiers. The processing and sample stages are separated in \cref{tab:sample-progression}; detailed definitions appear in \cref{app:sample-accounting}.

\begin{table}[ht]
\centering
\caption{Sample progression. Counts refer to different stages and are not interchangeable denominators.}
\label{tab:sample-progression}
\begin{tabularx}{\textwidth}{L{3.7cm}rY}
\toprule
Stage & Count & Interpretation \\
\midrule
Raw observed market identifiers & 110,828 & Distinct identifiers appearing anywhere in the 168-hour archive \\
Stratified enrichment candidates & 61,087 & Candidate markets processed under the locked day-stratified selection rule \\
Usable analysis sample & 13,298 & Markets passing metadata, activity, lifetime, and resolution-data adequacy gates \\
E3 replay denominator & 13,115 & Archived resolution-zone replay records after E3-specific eligibility filters; treated as the run denominator, not as a second estimate of the raw population \\
\bottomrule
\end{tabularx}
\end{table}

The earlier manuscript also reported a 17.4 percent ``resolution ceiling'' against a different candidate denominator and then compared it with the E3 count. That comparison obscured rather than clarified the sample. The statistic is not used in this revision's claims. The stage-specific counts in \cref{tab:sample-progression} are the only denominators used in the main text.

\subsection{Class composition}

The usable sample contains 408 politics markets, 1,518 crypto markets, 6,794 sports markets, and 4,578 markets classified as other. Sports therefore represent
\begin{equation}
  \frac{6{,}794}{13{,}298}=51.1\% 
\end{equation}
of the full usable sample. The previously reported 77.9 percent is the sports share within the three named classes politics, crypto, and sports:
\begin{equation}
  \frac{6{,}794}{408+1{,}518+6{,}794}=77.9\%.
\end{equation}
The denominator is stated whenever the 77.9 percent figure is used. The sample is sports-heavy under either description, but the full-sample and three-class shares answer different questions.

\subsection{Pre-specified analysis plan}

The estimators, gates, and materiality thresholds were recorded in an internally version-locked analysis appendix before the E1--E3 runs. They were not registered with an independent timestamped registry. This revision therefore uses \emph{pre-specified} rather than \emph{pre-registered}. The original plan and all deviations remain in the reproducibility record.

\subsection{Experiments}

The empirical design has three layers.

\begin{description}
  \item[E1: structural diagnostics.] Boundary-region depth and terminal-jump magnitude are measured against pre-specified floors. Additional microstructure facts are descriptive.
  \item[E2: fixed-path engine replay.] Three engine configurations are applied to observed paths under a position-agnostic grid, a deterministic position grid, and a synthetic-trader robustness layer.
  \item[E3: resolution-zone mechanics.] Four resolution mechanics isolate leverage compression, an experimental boundary-funding correction, and a staged halt.
\end{description}

The synthetic-trader layer is not a behavioral equilibrium model. It aggregates positions under stated rules while leaving the underlying path unchanged. Adding an adaptive behavioral simulation after observing the replay results would create a new model, assumptions, and estimands; it is reserved for a separately specified study rather than used as an ex post repair.

\section{Structural Non-Portability}
\label{sec:nonportability}

\subsection{Terminal account identity}

\begin{assumption}[Open exposure at resolution]
\label{ass:open-at-resolution}
The synthetic position remains open through $\tau$; no additional collateral is posted after entry; no pre-resolution close, deleveraging, or conversion occurs; and no external guarantee is counted inside the trader account.
\end{assumption}

\begin{proposition}[Adverse-outcome terminal equity]
\label{prop:terminal-equity}
Under \cref{ass:open-at-resolution}, consider a synthetic long of $x>0$ units entered at price $p_0\in(0,1)$ with collateral $C_0=xp_0/L$. If the terminal payout is $Y=0$, then
\begin{equation}
  E_\tau=-x p_0\left(1-\frac{1}{L}\right),
  \qquad
  S_\tau=x p_0\left(1-\frac{1}{L}\right).
  \label{eq:adverse-long-equity}
\end{equation}
Consequently, terminal account equity is negative for every $L>1$. The result is independent of the pre-resolution mark $q_{\tau^-}$ and of the path taken to it.
\end{proposition}

\begin{proof}
By \cref{eq:terminal-equity-general}, $E_\tau=C_0+x(Y-p_0)$. Substituting $C_0=xp_0/L$ and $Y=0$ gives \cref{eq:adverse-long-equity}. Equation \eqref{eq:telescoping-equity} shows that any pre-resolution mark-to-market gain or loss and the final mark change telescope to the same entry-to-terminal payoff.
\end{proof}

\begin{corollary}[Short-side symmetry]
A synthetic short of $x>0$ YES units entered at $p_0$, collateralized by $x(1-p_0)/L$, has adverse-outcome shortfall $x(1-p_0)(1-1/L)$ when $Y=1$, under the corresponding open-exposure assumptions.
\end{corollary}

\begin{remark}
The proposition does not say that every leveraged position must produce exchange loss. A sufficiently early liquidation, a top-up, a physically funded conversion, an offsetting hedge, or an external backstop can prevent or absorb the account shortfall. The result says that a position carried open to the adverse binary payoff cannot be made solvent by a backward-looking continuous-volatility estimate alone.
\end{remark}

\subsection{Implications for margin}

A continuous-volatility margin term may estimate ordinary path variation. It cannot alter the worst-case terminal account identity. To admit $L>1$ while avoiding the shortfall in \cref{prop:terminal-equity}, the mechanism must enforce at least one of the following before resolution:

\begin{enumerate}[label=(\alph*)]
  \item collateral top-ups that move effective leverage toward one;
  \item position reduction or full close;
  \item conversion into fully funded outcome claims;
  \item an enforceable external guarantee or loss-bearing capital layer; or
  \item a contract payoff that does not preserve the adverse binary exposure through resolution.
\end{enumerate}

A dated future has the same issue at expiry if it remains under-collateralized. The underlying problem is finite terminal settlement, not the label ``perpetual.'' Binary support makes the adverse loss exact and common; it is not the only possible setting in which terminal collateral insufficiency arises.

\subsection{Boundary funding trilemma}

Let $I\in(0,1)$ be the event index and $q\in[0,1]$ the derivative mark. Let $f(q,I)$ be the per-unit funding transfer over one funding interval, positive from longs to shorts under a fixed sign convention.

\begin{assumption}[Funding feasibility]
\label{ass:funding-feasibility}
Funding is bounded by a finite transfer cap $\bar f$ determined by collateral, payer solvency, and participation constraints:
\begin{equation}
  |f(q,I)|\leq \bar f <\infty
  \quad\text{for all }(q,I)\in[0,1]\times(0,1).
\end{equation}
\end{assumption}

\begin{proposition}[Boundary funding trilemma]
\label{prop:funding-trilemma}
No funding rule satisfying \cref{ass:funding-feasibility} can impose a uniform lower bound proportional to the relative basis
\begin{equation}
  |f(q,I)|\geq \kappa\frac{|q-I|}{\min\{I,1-I\}}
  \label{eq:relative-pressure}
\end{equation}
for all $q\in[0,1]$, $I\in(0,1)$, and any fixed $\kappa>0$. Conversely, a basis-only rule $f=c(q-I)$ becomes arbitrarily small for a fixed relative deviation as $I\downarrow0$ or $I\uparrow1$.
\end{proposition}

\begin{proof}
For the first statement, fix any $q\in(0,1)$ and let $I\downarrow0$. The right-hand side of \eqref{eq:relative-pressure} diverges, contradicting the finite cap $\bar f$. For the second statement, fix $\rho>0$ and set $q=(1+\rho)I$ for sufficiently small $I$. Then $|q-I|=\rho I$ and $|c(q-I)|=c\rho I\to0$, while the relative deviation remains $\rho$.
\end{proof}

The proposition identifies a design trilemma. A mechanism cannot simultaneously guarantee uniform relative-basis pressure at every boundary point, keep transfers bounded, and preserve participation of the paying side. A feasible rule must use a denominator floor, a transfer cap, boundary-specific leverage/listing restrictions, or some combination. The uncapped correction in the original paper satisfied the mathematical pressure objective by abandoning economic feasibility. It is not retained as a deployment recommendation.

\section{Candidate Resolution-Aware Controls}
\label{sec:candidate-controls}

The empirical replay evaluates a family of candidate controls. This section states their intended risk channels and their limitations. The family is not presented as a solved architecture.

\subsection{Robust reference index}

A composite index can combine midpoint, depth-weighted midpoint, and recent executed prices. Its purpose is to reduce sensitivity to any one noisy or manipulable input. It cannot create executable liquidity or determine the true probability of the event. The replay uses the existing archived implementation and does not claim optimal index weights.

\subsection{Margin decomposition}

A useful accounting decomposition is
\begin{equation}
  M_t^{\mathrm{maint}}
  =M_t^{\mathrm{path}}+M_t^{\mathrm{terminal}}+M_t^{\mathrm{execution}},
  \label{eq:margin-decomposition}
\end{equation}
where the terms cover ordinary path variation, adverse terminal settlement, and the cost of reducing exposure in available depth. The decomposition avoids asking one realized-volatility multiplier to perform three jobs.

The terminal term does not by itself solve \cref{prop:terminal-equity}: if it is large enough to cover the full adverse payoff, effective leverage approaches one. If it rises gradually, it can liquidate positions that would later recover. The empirical question is therefore not whether terminal risk exists, but whether the schedule improves exchange loss enough to justify pre-emption.

\subsection{Leverage compression}

A time- or state-dependent cap $L_{\max}(t)$ reduces new risk as resolution approaches. Calendar time is an imperfect trigger when event timing is uncertain. A production rule would preferably use observable state transitions, data quality, and executable depth in addition to the scheduled event time. The present replay uses the archived time-to-resolution schedule.

\subsection{Economically bounded funding}

An admissible boundary-aware diagnostic can be written as
\begin{equation}
  f_t=\operatorname{clip}_{[-\bar f_t,\bar f_t]}
  \left[
    \alpha_t(q_t-I_t)
    +\beta_t\mathbf{1}_{\mathcal B_t}
    \frac{q_t-I_t}{\max\{\varepsilon,\min(I_t,1-I_t)\}}
  \right],
  \label{eq:capped-funding}
\end{equation}
where $\varepsilon>0$ floors the denominator and $\bar f_t$ enforces collateral and participation limits. A further account constraint can require that the interval transfer not exceed a stated fraction of payer equity.

Equation \eqref{eq:capped-funding} is a feasible design family, not the rule evaluated by the archived replay. The replay used an uncapped experimental correction. The correction had little incremental effect on the observed E3 paths, and the revision does not infer performance for the capped family without a new run.

\subsection{Resolution-zone halt}

A staged halt or reduce-only transition can prevent new positions and in-flight liquidations after a cutoff. It addresses an execution window. It does not change the terminal payoff of positions that remain open. The halt should therefore be evaluated separately from terminal account shortfall.

\subsection{Risk-channel map}

\begin{table}[ht]
\centering
\caption{Candidate controls and the risk channels they can and cannot address.}
\label{tab:risk-channel-map}
\small
\begin{tabularx}{\textwidth}{L{2.8cm}YYY}
\toprule
Control & Primary channel & What it may improve & What it cannot establish \\
\midrule
Composite index & Index measurement & More robust mark input & Executable depth or terminal solvency \\
Jump-aware margin & Account collateral & Earlier recognition of terminal exposure & Low-pre-emption calibration or equilibrium demand \\
Leverage compression & Exposure admission & Smaller positions near finality & Settlement of already open under-collateralized positions \\
Capped boundary funding & Mark--index carry & Penalizes some relative deviations & Uniform boundary convergence under bounded transfers \\
Staged halt & Execution timing & Removes a late trading/liquidation window & Adverse terminal payoff or account shortfall \\
\bottomrule
\end{tabularx}
\end{table}

A separate listing screen based on data quality, rule clarity, and liquidity can exclude markets for which no control is credible. The screen is retained in the technical supplement but is not a headline contribution of this revision.

\section{Empirical Results}
\label{sec:results}

\subsection{Structural diagnostics}

Two pre-specified diagnostics pass.

\paragraph{Boundary depth.} The pooled median boundary-depth ratio is 1.72 against a floor of 1.5. The class medians reported by the archived run are 1.69 for crypto, 1.75 for politics, 1.74 for sports, and 1.70 for other.

\paragraph{Terminal jump.} Among markets with usable final-hour observations, the median absolute final-hour move is 0.50 against a floor of 0.10. A separate cohort has no observable final-hour price or book update, which is treated as final-hour illiquidity rather than imputed price continuity.

These results establish that the candidate engine is being tested against material boundary and terminal phenomena. They do not establish that the proposed controls solve them.

\subsection{E2: fixed-path engine replay}

The configurations are:

\begin{description}
  \item[$\mathsf{C}_0$] static continuous-volatility margin, basis-only funding, fixed leverage cap;
  \item[$\mathsf{C}_1$] jump-aware margin and leverage compression;
  \item[$\mathsf{C}_2$] $\mathsf{C}_1$ plus the experimental funding correction and staged resolution protocol.
\end{description}

\begin{table}[ht]
\centering
\caption{Maintenance-path pass rate: fraction of market--side grid cells that do not breach maintenance margin on the observed path.}
\label{tab:e2-survival-revised}
\begin{tabular}{lccccc}
\toprule
Configuration & $L=1$ & $L=2$ & $L=3$ & $L=5$ & $L=10$ \\
\midrule
$\mathsf{C}_0$ & 33.5\% & 31.1\% & 30.1\% & 28.8\% & 27.3\% \\
$\mathsf{C}_1$ & 28.7\% & 27.1\% & 26.5\% & 25.7\% & 24.6\% \\
$\mathsf{C}_2$ & 27.9\% & 26.3\% & 25.5\% & 24.5\% & 23.2\% \\
\bottomrule
\end{tabular}
\end{table}

At $L=5$, $\mathsf{C}_2$ has a lower maintenance-path pass rate and approximately 6 percent higher liquidation incidence than $\mathsf{C}_0$, rather than the pre-specified 30 percent reduction. The result is a fixed-path pre-emption effect: the dynamic rules trigger before some paths recover. It is not evidence that the same traders would have entered or maintained the same positions under $C_2$.

\begin{table}[ht]
\centering
\caption{Hypothetical insurance-pool drawdown on the deterministic position grid. Values are archived run outputs, not forecasts of a deployed exchange.}
\label{tab:e2-drawdown-revised}
\begin{tabular}{lrrrr}
\toprule
Class & $\mathsf{C}_0$ & $\mathsf{C}_1$ & $\mathsf{C}_2$ & $\mathsf{C}_2$ versus $\mathsf{C}_0$ \\
\midrule
Crypto   & 68.3 M  & 60.8 M  & 60.7 M  & $-11.1\%$ \\
Politics & 7.2 M   & 7.0 M   & 7.0 M   & $-2.8\%$ \\
Sports   & 268.6 M & 255.3 M & 254.9 M & $-5.1\%$ \\
Other    & 201.0 M & 198.3 M & 194.4 M & $-3.3\%$ \\
\midrule
All      & 545.1 M & 521.3 M & 517.1 M & $-5.1\%$ \\
\bottomrule
\end{tabular}
\end{table}

The pooled reduction is 5.1 percent against a pre-specified 50 percent floor. The crypto class has the largest reduction, but the single-week and class-composition limits preclude a general class ranking.

The synthetic-trader robustness layer improves median profit and loss from $-70$ to $-60$ but worsens the mean from $-122$ to $-144$. These are model-dependent outputs, not welfare estimates. The median-direction test passes; the controlling E2 materiality tests fail.

\subsection{E3: resolution-zone mechanics}

The mechanics are $M_0$ forced expiry, $M_1$ leverage compression, $M_2$ compression plus the experimental funding correction, and $M_3$ the full staged halt.

\begin{table}[ht]
\centering
\caption{Final-hour liquidation incidence by mechanic. The reduction under $M_3$ is largely a consequence of removing the trading window.}
\label{tab:e3-liquidation-revised}
\begin{tabular}{lcccc}
\toprule
Mechanic & $L=2$ & $L=3$ & $L=5$ & $L=10$ \\
\midrule
$M_0$ & 0.115\% & 0.119\% & 0.107\% & 0.128\% \\
$M_1$ & 0.128\% & 0.128\% & 0.115\% & 0.123\% \\
$M_2$ & 0.128\% & 0.128\% & 0.115\% & 0.123\% \\
$M_3$ & 0.025\% & 0.025\% & 0.021\% & 0.021\% \\
\bottomrule
\end{tabular}
\end{table}

\begin{table}[ht]
\centering
\caption{Terminal account-shortfall incidence on the E3 grid. The original paper called this bad-debt frequency; exchange bad debt would require a full recovery waterfall.}
\label{tab:e3-shortfall-revised}
\begin{tabular}{lcccc}
\toprule
Mechanic & $L=2$ & $L=3$ & $L=5$ & $L=10$ \\
\midrule
$M_0$ & 51.7\% & 53.3\% & 54.8\% & 56.4\% \\
$M_1$ & 51.6\% & 53.3\% & 54.8\% & 56.4\% \\
$M_2$ & 52.5\% & 54.0\% & 55.6\% & 57.5\% \\
$M_3$ & 53.0\% & 54.5\% & 56.0\% & 57.9\% \\
\bottomrule
\end{tabular}
\end{table}

$M_3$ lowers final-hour liquidation incidence by 80.4 percent relative to $M_0$, passing the 50 percent floor. The result is mechanical because the halt removes the interval in which those liquidations can occur. Terminal account-shortfall incidence rises by 2.4 percent rather than falling by the pre-specified 75 percent. The halt therefore changes the execution path without repairing the terminal balance sheet.

$M_1$ and $M_2$ are nearly identical in final-hour liquidation incidence. The experimental boundary-funding component did not contribute materially on these paths, so the replay provides no evidence that an uncapped or capped boundary correction is economically effective.

\subsection{Test summary}

\begin{table}[ht]
\centering
\caption{Pre-specified tests. Structural diagnostics are separated from the five E2/E3 materiality tests.}
\label{tab:test-summary-revised}
\begin{tabularx}{\textwidth}{L{3.7cm}L{3.6cm}YL{1.4cm}}
\toprule
Test & Pre-specified criterion & Archived result & Status \\
\midrule
Boundary-depth diagnostic & median $\geq1.5$ & 1.72 pooled & Pass \\
Terminal-jump diagnostic & median $\geq0.10$ & 0.50 pooled & Pass \\
E2 liquidation incidence & $\geq30\%$ reduction at $L=5$ & approximately 6\% worse & Fail \\
E2 grid drawdown & $\geq50\%$ pooled reduction & 5.1\% reduction & Fail \\
E2 synthetic median P\&L & non-negative direction & 14\% median improvement & Pass$^{a}$ \\
E3 final-hour liquidations & $\geq50\%$ reduction & 80.4\% reduction, by halt construction & Pass$^{b}$ \\
E3 terminal shortfall & $\geq75\%$ reduction & 2.4\% worse & Fail \\
\bottomrule
\end{tabularx}
\begin{flushleft}\footnotesize
$^{a}$Synthetic robustness layer; not a controlling empirical test. $^{b}$Mechanically induced by removal of the final-hour trading window.
\end{flushleft}
\end{table}

Three of the five E2/E3 materiality tests fail. The evidence supports non-portability and negative design lessons, not validation of the candidate control family.

\section{Interpretation, Design Implications, and Limitations}
\label{sec:interpretation}

\subsection{What the replay identifies}

The fixed-path results distinguish three channels.

\paragraph{Terminal balance-sheet channel.} If under-collateralized exposure survives to the adverse outcome, \cref{prop:terminal-equity} determines the account shortfall. A halt does not change this identity.

\paragraph{Pre-emption channel.} Dynamic margin and leverage compression can force exits before ordinary path recovery. The replay shows this channel on the observed paths. It does not prove that every calibration or every venue produces the same ordering.

\paragraph{Execution-window channel.} A halt can prevent in-flight liquidations after the halt time. The 80.4 percent E3 reduction measures this channel, but it is not evidence of lower terminal loss.

These channels should not be aggregated into one statement that the framework ``works'' or ``fails.'' Each control has a narrower object.

\subsection{Concrete design directions}

The negative results suggest four research directions.

\begin{enumerate}
  \item \textbf{Pre-finality deleveraging or physical conversion.} If the terminal account identity is the dominant loss channel, the strongest intervention is to extinguish or transform the synthetic obligation before finality, not merely to halt trading around it.
  \item \textbf{State-triggered rather than continuously reactive margin.} Margin changes tied to observable resolution states, rule clarifications, disputes, or liquidity failures may avoid some volatility-driven pre-emption. This is a hypothesis requiring a new replay, not a conclusion from the current one.
  \item \textbf{Bounded funding plus eligibility restrictions.} Funding transfers should be capped and subject to payer-equity limits. Boundary states in which the cap makes relative-basis control ineffective should trigger leverage reduction, reduce-only operation, or ineligibility.
  \item \textbf{Position and class calibration.} The deterministic grid uses positions far larger than typical fills in parts of the sample. Future tests should use observed position distributions where available and should report event-class results without treating one week as stable population calibration.
\end{enumerate}

These directions do not rescue the present framework ex post. They define separately testable mechanisms.

\subsection{Why no adaptive equilibrium simulation is added here}

A behavioral simulation could allow traders and liquidity suppliers to respond to margin changes and halts. Such a simulation would require assumptions about entry, leverage demand, market-making, arbitrage capital, order cancellation, and strategic behavior. The current dataset does not identify those response functions. Adding one after seeing the negative replay would replace a transparent fixed-path stress test with an author-specified model. We therefore state the equilibrium limitation directly and reserve adaptive-agent analysis for a separately specified experiment with its own hypotheses and sensitivity design.

\subsection{Data and external-validity limits}

The evidence covers one seven-day Polymarket window and is sports-heavy: sports account for 51.1 percent of the full usable sample and 77.9 percent of the three named classes. Results are platform- and window-specific. The 2024 election cycle is outside the archive window. Address-level trader positions and endogenous derivative open interest are unavailable. The replay treats the spot-market path as fixed. The synthetic-trader layer is a robustness exercise, not observed behavior.

The terminal-jump diagnostic is observed only where a usable final-hour quote exists. A substantial cohort becomes inactive before resolution. This missingness is economically meaningful, but it prevents unconditional estimation of the jump distribution across all resolved markets.

\subsection{Formal scope}

The terminal-equity proposition assumes no top-up, close, conversion, or guarantee before resolution. It is not a theorem that leverage is impossible. It identifies what any viable architecture must add.

The funding proposition assumes a finite payment cap. It does not show that every capped rule is ineffective; it shows that uniform relative-basis control arbitrarily close to the boundaries is incompatible with bounded transfers.

The risk-control definitions are mechanism candidates. The existing replay does not validate the capped funding family in \cref{eq:capped-funding}, state-triggered margin, physical conversion, or equilibrium deployment.

\subsection{Relationship to real-valued terminal variables}

The account identity extends to any terminal payoff $Y$ for which the adverse entry-to-terminal loss exceeds posted collateral. A real-valued prediction target can therefore reproduce the same solvency issue if the derivative has a finite settlement time, the payoff can move discontinuously, or the market closes before the position can be reduced. The binary case is distinctive because the payoff support is exactly known, adverse loss is easy to express, and the normalized index creates the boundary-funding trilemma. The broader lesson is about terminal maturity and market closure; the unit interval adds a separate geometry.

\section{Conclusion}
\label{sec:conclusion}

A crypto-style perpetual overlay does not remove the terminal economics of the event claim it references. If a synthetic position remains open through an adverse binary payoff, the total entry-to-terminal loss determines account solvency. The corrected account identity shows that, absent top-ups, close, conversion, or external support, any leverage above one leaves the adverse-outcome account under-collateralized. The pre-resolution mark changes when the problem appears, not the terminal amount.

Funding has a separate boundary limitation. Basis-only transfers become negligible for fixed relative errors near zero or one. A relative-basis correction can avoid that vanishing incentive only by growing without bound unless it is clipped or its denominator is floored. Once payer solvency and participation are imposed, funding cannot provide uniform boundary control by itself.

The Polymarket replay supports the structural concerns but not the candidate control family as a solution. Boundary depth and terminal jumps are economically material. Dynamic margin and leverage compression increase liquidation incidence on the observed paths; pooled hypothetical drawdown falls only modestly; and a staged halt removes final-hour liquidations while leaving terminal account-shortfall incidence slightly worse. Three of five pre-specified E2/E3 materiality tests fail.

The paper's positive contribution is therefore a negative-design result: it identifies which balance-sheet and mechanism assumptions make a naive transplant fail, supplies a reproducible fixed-path stress-test protocol, and narrows the set of credible follow-up designs. Future work should test pre-finality deleveraging, physically funded conversion, state-triggered margin, and bounded funding under explicit participation constraints. None is validated by the present replay.

\appendix
\section{Sample Accounting and Denominator Discipline}
\label{app:sample-accounting}

\subsection{Window}

The data window is 2026-04-21 00:00 UTC through 2026-04-27 23:59 UTC. These dates run from Tuesday through Monday. The window contains 168 hourly files.

\subsection{Stage definitions}

\begin{table}[ht]
\centering
\caption{Denominator registry for the revised paper.}
\label{tab:denominator-registry}
\begin{tabularx}{\textwidth}{L{3.2cm}rYY}
\toprule
Identifier & Count & Inclusion rule & Permitted use \\
\midrule
Raw archive markets & 110,828 & Any distinct market identifier observed in the window & Archive scale only \\
Stratified candidates & 61,087 & Candidate set selected and enriched under the locked day-stratified rule & Sample progression and join diagnostics \\
Usable analysis markets & 13,298 & Pass lifetime, activity, metadata, and resolution-data adequacy gates & E1 and E2 denominator \\
E3 replay records & 13,115 & Pass E3-specific resolution-zone filters in the archived run & E3 tables only \\
Three named classes & 8,720 & Politics, crypto, or sports among usable markets & Denominator for the 77.9\% sports-share trigger \\
Full usable class set & 13,298 & Named classes plus other & Denominator for the 51.1\% full-sample sports share \\
\bottomrule
\end{tabularx}
\end{table}

The revision does not use the prior 17.4 percent resolution-ceiling statistic. The source package did not contain the run-level key table needed to reconcile that statistic with the E3 count without re-executing the pipeline, and the statistic is not necessary for any main estimand.

\subsection{Class counts}

\begin{table}[ht]
\centering
\caption{Usable analysis sample by class.}
\label{tab:class-counts-revised}
\begin{tabular}{lrr}
\toprule
Class & Markets & Share of full usable sample \\
\midrule
Politics & 408 & 3.07\% \\
Crypto & 1,518 & 11.42\% \\
Sports & 6,794 & 51.09\% \\
Other & 4,578 & 34.43\% \\
\midrule
Total & 13,298 & 100.00\% \\
\bottomrule
\end{tabular}
\end{table}

The sports-share trigger was defined on politics, crypto, and sports only. Its denominator is therefore 8,720, giving 77.9 percent. This trigger is retained as an analysis-plan consequence rule, while the full-sample composition is reported separately.

\section{Pre-Specified Estimands and Thresholds}
\label{app:prespecified}

The analysis plan was frozen internally before the E1--E3 runs but was not placed in an independent timestamped registry. The correct description is \emph{pre-specified} or \emph{version-locked}, not externally pre-registered.

\subsection{Structural diagnostics}

\begin{enumerate}
  \item Boundary-depth ratio: pooled median at least 1.5.
  \item Absolute terminal jump over the final hour: pooled median at least 0.10 among markets with usable terminal observations.
\end{enumerate}

\subsection{E2 materiality tests}

\begin{enumerate}
  \item At $L=5$, the full configuration reduces liquidation incidence by at least 30 percent relative to the baseline.
  \item The full configuration reduces hypothetical insurance-pool drawdown by at least 50 percent on the deterministic position grid.
  \item The synthetic-trader median profit-and-loss direction is weakly better under the full configuration. This is a robustness test, not a controlling empirical test.
\end{enumerate}

\subsection{E3 materiality tests}

\begin{enumerate}
  \item The staged mechanic reduces final-hour liquidation incidence by at least 50 percent relative to forced expiry.
  \item The staged mechanic reduces terminal account-shortfall incidence by at least 75 percent.
\end{enumerate}

The original appendix used the term ``bad debt'' for the second metric. The revision uses terminal account shortfall because the replay does not model the full exchange recovery waterfall.

\section{Additional Empirical Detail}
\label{app:detailed-results}

\subsection{Descriptive microstructure facts}

The archived analysis reports the following additional descriptive facts. They are useful for mechanism design but are not independent hypothesis tests.

\begin{itemize}
  \item Median absolute basis is 0.013 in selected news windows and 0.037 in matched control windows under the fast index estimator.
  \item Median half-spread is approximately 0.0055 in the low boundary region, 0.2700 in the mid region, and 0.00525 in the high boundary region.
  \item Median/mean/p99 trade sizes are 20/713/10,000 USDC in politics, 16.5/301/2,500 in sports, 10/79/500 in crypto, and 7.31/43.5/500 in other.
  \item Final-24-hour activity relative to the earlier baseline is 24.6 times in crypto, 2.5 times in sports, 0.9 times in other, and 0.68 times in politics.
\end{itemize}

These values are sample- and estimator-specific. They should not be treated as stable venue parameters without replication.

\subsection{Synthetic-trader robustness output}

\begin{table}[ht]
\centering
\caption{Synthetic-trader profit-and-loss output retained as robustness evidence only.}
\label{tab:e2c-appendix}
\begin{tabular}{lrrr}
\toprule
Configuration & p10 & Median & Mean \\
\midrule
$\mathsf{C}_0$ & $-805$ & $-70$ & $-122$ \\
$\mathsf{C}_2$ & $-800$ & $-60$ & $-144$ \\
\bottomrule
\end{tabular}
\end{table}

The median improves while the mean worsens. The output cannot be interpreted as welfare because the trader population and entry rules are synthetic and the market path is fixed.

\subsection{Archived numerical provenance}

The empirical values in this revision are copied from the archived publication outputs; no new empirical run was performed. Run identifiers and machine-level provenance remain in the repository manifests rather than in the paper narrative.

\section{Reproducibility and Data Release}
\label{app:reproducibility}

The empirical pipeline separates raw ingest, cleaning, feature construction, and deterministic engine evaluation. Every released artifact should record the code commit, input hashes, parameters, seeds where randomness enters, row counts, and output hashes.

The companion release plan includes:

\begin{itemize}
  \item \texttt{pmxt-stylized-facts-v1} for the structural diagnostics;
  \item \texttt{pmxt-counterfactual-replay-v1} for E2/E3 outputs;
  \item schemas, methodology, known limitations, processing ledger, descriptive statistics, and cryptographic manifests;
  \item reconstruction instructions for source data that cannot be redistributed;
  \item Zenodo archival versions and ForesightFlow dataset pages after final audit.
\end{itemize}

The revision itself performs no new calculation. It corrects the formal account identity, separates account shortfall from exchange bad debt, changes the interpretation of the funding rule, reconciles narrative denominators, and rewrites the empirical conclusions around the archived outputs.

\subsection{Mechanical-replay limitation}

Reproduction of $g_c(X^{\mathrm{obs}})$ does not validate the counterfactual market path $X(c)$. A complete deployment evaluation would require either observed operation of the product or a separately specified behavioral model with independently justified response functions.

\section{Notation and Terminology}
\label{app:notation}

\begin{table}[ht]
\centering
\caption{Canonical notation used in the revised paper.}
\label{tab:notation-revised}
\small
\begin{tabularx}{\textwidth}{L{2.1cm}Y}
\toprule
Symbol & Meaning \\
\midrule
$Y$ & Terminal payout of the YES claim, in $\{0,1\}$ \\
$\tau$ & Contract terminal timestamp used by the replay; not a claim that all oracle, protocol, and redemption clocks coincide \\
$p_t$ & Reference event-market price \\
$I_t$ & Engine reference index \\
$q_t$ & Derivative mark price \\
$x$ & Position size in claim units \\
$\mathsf{C}_0$ & Collateral posted at entry \\
$L$ & Premium-based leverage for the position \\
$E_t$ & Account equity \\
$S_\tau$ & Terminal account shortfall \\
$M_t$ & Margin requirement \\
$f_t$ & Per-unit funding transfer for one interval \\
$\mathsf{C}_0,\mathsf{C}_1,\mathsf{C}_2$ & E2 engine configurations \\
$M_0,\ldots,M_3$ & E3 resolution-zone mechanics \\
\bottomrule
\end{tabularx}
\end{table}

\paragraph{Central limit order book.} The abbreviation CLOB is expanded at first use. The phrase ``CLOB book'' is avoided because the word ``book'' is already contained in the abbreviation.

\paragraph{Account shortfall versus bad debt.} Account shortfall is negative trader-account equity after settlement. Exchange bad debt is the residual loss after liquidation, collateral recovery, guarantees, and insurance. The current replay identifies the former more directly than the latter.

\paragraph{Pre-specified versus pre-registered.} Pre-specified means the analysis plan was fixed before the run. Pre-registered is reserved for an independently timestamped registry.

\printbibliography[heading=bibintoc,title={References}]

@article{wolfers_zitzewitz_2004,
  author  = {Wolfers, Justin and Zitzewitz, Eric},
  title   = {Prediction Markets},
  journal = {Journal of Economic Perspectives},
  year    = {2004},
  volume  = {18},
  number  = {2},
  pages   = {107--126},
  doi     = {10.1257/0895330041371321}
}

@article{manski_2006,
  author  = {Manski, Charles F.},
  title   = {Interpreting the Predictions of Prediction Markets},
  journal = {Economics Letters},
  year    = {2006},
  volume  = {91},
  number  = {3},
  pages   = {425--429},
  doi     = {10.1016/j.econlet.2006.01.004}
}

@article{he_manela_schwert_2024_perpetuals,
  author  = {He, Songrun and Manela, Asaf and Schwert, Michael},
  title   = {Fundamentals of Perpetual Futures},
  journal = {arXiv preprint},
  year    = {2024},
  eprint  = {2212.06888},
  archivePrefix = {arXiv},
  url     = {https://arxiv.org/abs/2212.06888},
  note    = {Canonical mechanism paper for perpetual futures: derives no-arbitrage relationship between perpetual price, spot price, and funding rate. Useful for §2.2 and §6 baseline framing.}
}

@article{dai_li_yang_2025_arbitrage_perpetuals,
  author  = {Dai, Min and Li, Linfeng and Yang, Chen},
  title   = {Arbitrage in Perpetual Contracts},
  year    = {2025},
  note    = {SSRN working paper (abstract 5262988). Derives model-free no-arbitrage bounds for perpetual futures prices, examining the role of the funding clamping function. Relevant to §6 and §7.4 funding-rule discussion.}
}

@article{zhang_2026_funding_rate_mechanism,
  author  = {Zhang, Tianyang},
  title   = {Funding Rate Mechanism in Perpetual Futures},
  year    = {2026},
  month   = {February},
  doi     = {10.2139/ssrn.6185958},
  note    = {SSRN working paper (abstract 6185958). Studies funding rate as algorithmic feedback rule in continuous-time equilibrium with risk-constrained arbitrageurs and momentum speculators. Derives stability condition; analyzes funding caps, clamps, jump-and-crisis extension. Highly relevant to §7.4 funding rule discussion and §8.2 E2 evaluation.}
}

@article{dubach_2026_polymarket_anatomy,
  author  = {Dubach, Philipp D.},
  title   = {The Anatomy of a Decentralized Prediction Market: Microstructure Evidence from the Polymarket Order Book},
  year    = {2026},
  eprint  = {2604.24366},
  archivePrefix = {arXiv},
  url     = {https://arxiv.org/abs/2604.24366},
  note    = {Tick-level Polymarket WebSocket archive joined to on-chain trade record (30 billion events, 52 days, 600-market panel). Eight stylized facts: longshot spread premium, depth-concentration profile (geometric grid not top-of-book), null block-clock alignment, maker-wallet diversity, category-conditional effective spread, sub-50ms ingestion latency, low self-counterparty wash. Highly relevant baseline for §3.3, §5 stylized facts, §2.1 related work.}
}

@article{tsang_yang_2026_polymarket_anatomy,
  author  = {Tsang, Kwok Ping and Yang, Zichao},
  title   = {The Anatomy of Polymarket: Evidence from the 2024 Presidential Election},
  year    = {2026},
  month   = {March},
  eprint  = {2603.03136},
  archivePrefix = {arXiv},
  primaryClass = {econ.GN},
  url     = {https://arxiv.org/abs/2603.03136},
  note    = {Transaction-level on-chain analysis of Polymarket's 2024 US Presidential Election market. Develops volume decomposition (exchange-equivalent volume, net inflow, gross market activity) for blockchain-based prediction markets where minting/burning/conversion mix with conventional exchange. Documents three episodes (Biden withdrawal, September debate, October whales). Cited from §1 motivation as anchor for "substantial publicly reported volume around the 2024 US election cycle".}
}

@article{rahman_alchami_clark_2025_sok_depms,
  author  = {Rahman, Nahid and Al-Chami, Joseph and Clark, Jeremy},
  title   = {{SoK}: Market Microstructure for Decentralized Prediction Markets ({DePMs})},
  year    = {2025},
  month   = {October},
  eprint  = {2510.15612},
  archivePrefix = {arXiv},
  primaryClass = {cs.CE},
  url     = {https://arxiv.org/abs/2510.15612},
  note    = {Systematization of knowledge: history of DePM proposals from 2011 onward, modular workflow (underlying infrastructure, market topic, share structure, trading, resolution, settlement, archiving), design variants and trade-offs around decentralization, expressiveness, manipulation resistance. Argues Polymarket succeeded by deviating from earlier Augur/Truthcoin designs. Useful for §2.1 survey paragraph and §3.3 background on DePM landscape. Cross-listed q-fin.TR and cs.CR.}
}
\end{document}